\documentclass[journal,10pt]{IEEEtran} 
\usepackage[figurename=Fig.]{caption}
\usepackage{tikz}
\usetikzlibrary{calc, positioning, shapes.geometric, decorations.pathreplacing}
\usepackage{array}
\usepackage{verbatim}
\usetikzlibrary{calc, shapes, arrows, positioning}
\usepackage{authblk}
\usepackage{blindtext}
\usepackage{ifpdf}
\usepackage{graphicx}
\usepackage{tabulary}
\usepackage{amsmath,amsthm,amssymb}
\usepackage{amsmath}
\usepackage{kantlipsum}
\allowdisplaybreaks
\usepackage{cleveref}
\usepackage{lipsum}%
\usepackage{optidef}
\usepackage[english]{babel} 
\theoremstyle{plain}
\newtheorem{theorem}{Theorem}

\crefname{theorem}{Theorem}{theorem}
\crefname{lemma}{Lemma}{Lemmas}

\usepackage{cite}
\usepackage{dsfont}
\usepackage[justification=centering]{caption}
\usepackage{color}
\usepackage{algorithm}
\usepackage{algorithmicx, algpseudocode}
\usepackage{etoolbox}
\usepackage{subcaption}
\usepackage{balance}
\usepackage{empheq,etoolbox}
\usepackage[utf8]{inputenc}
\usepackage{multirow}
\usepackage{float}

\usepackage{amssymb}
\usepackage{pifont}
\usepackage[hmargin=1.27cm,top=1.3cm,bottom=1.3cm]{geometry}

\usepackage{physics}

\usepackage{babel}
\floatstyle{plaintop}
\restylefloat{table}
\patchcmd{\subequations}
\usepackage{lipsum} 
\allowdisplaybreaks

\tikzset{brace/.style={decorate, decoration={brace}},
 brace mirrored/.style={decorate, decoration={brace,mirror}},
}

\newcounter{brace}
\newcounter{arrow}
\begin{document}
 \captionsetup[figure]{name={Fig.},labelsep=period}


\title{{\color{black}Hybrid Wireless-Fed Pinching-Antenna Systems with Residual Self-Interference-Aware Optimization}}

\bstctlcite{IEEEexample:BSTcontrol}
\author{Hao Feng, Ming Zeng, Ebrahim Bedeer, Xingwang Li, Octavia A. Dobre, \textit{Fellow, IEEE} and Zhiguo Ding, \textit{Fellow, IEEE}
    \thanks{This work was supported in part by NSERC under its Discovery Grant. (Corresponding author: Ming Zeng.)}
    \thanks{H. Feng is with Hunan Institute of Engineering, Xiangtan, China, and also with Donghua University, Shanghai, China (e-mail: 1219001@mail.dhu.edu.cn).}
    
    \thanks{M. Zeng is with Laval University, Quebec City, Canada (email: ming.zeng@gel.ulaval.ca).}

    \thanks{E. Bedeer is with University of Saskatchewan, Saskatoon, SK, Canada (email: e.bedeer@usask.ca).}

    \thanks{X. Li is with Henan Polytechnic University, Jiaozuo, China (email: lixingwang@hpu.edu.cn).}





    






\thanks{O. A. Dobre is with Memorial University, St. John’s, Canada (e-mail: odobre@mun.ca).}



\thanks{Z. Ding is with Nanyang Technological University, Singapore. (e-mail:zhiguo.ding@ieee.org).}


    }

\maketitle

\begin{abstract}
Pinching-antenna systems (PASS) have recently emerged as a promising solution for enhancing coverage in high-frequency wireless communications by guiding signals through dielectric waveguides and radiating them via position-adjustable antennas. However, their practical deployment is limited by waveguide attenuation and the need for physical line installation, which restrict flexibility and coverage extension. {\color{black}To address these challenges, this paper proposes a hybrid wireless-fed PASS architecture, where a base station equipped with an antenna array provides adaptive directional transmission to a full-duplex amplify-and-forward relay employing a horn antenna to feed the waveguide. This hybrid design balances beamforming flexibility and low-complexity directional waveguide interfacing.
Residual self-interference (SI) at the full-duplex relay is explicitly modeled to capture practical system impairments.} Under this framework, a total power minimization problem is formulated subject to a quality-of-service constraint at the user equipment, involving the joint optimization of the pinching-antenna position, the relay amplification gain, and the base station transmit power. By exploiting the structure of the end-to-end signal-to-noise ratio, the optimal pinching-antenna position is first obtained in closed form by balancing waveguide attenuation and free-space path loss. Closed-form expressions for the optimal relay gain and transmit power are then derived. 
{\color{black}Numerical results under the adopted system-level model demonstrate that the proposed scheme reduces total power consumption compared with conventional benchmark systems,} while providing a more realistic and robust design by accounting for residual SI.
\end{abstract}

\IEEEpeerreviewmaketitle

\section{Introduction}
High-frequency wireless communication, including millimeter-wave (mmWave) and terahertz (THz) systems, has emerged as a key enabler for next-generation networks due to the availability of large bandwidths \cite{Sun_TVT18, Hao_Network22}. However, severe path loss and susceptibility to blockage significantly limit coverage at these frequencies. To address this challenge, pinching-antenna systems have recently been proposed as an effective architecture for enhancing signal propagation in high-frequency regimes \cite{Atsushi_22}. In such systems, the transmitted signal is guided through a dielectric waveguide and radiated by a pinching antenna whose position along the waveguide can be dynamically adjusted \cite{zeng2025_WCM, liu2025pinching, yang2025, ding2024}. This architecture enables flexible control of radiation characteristics, improves spatial adaptability, and offers a low-loss alternative to purely over-the-air transmission, making it particularly attractive for high-frequency applications \cite{ Zeng_COMML25, Xiao_COMML25, Zhao_TCOM25, Ouyang_COMML25, zeng2025EE, fu2025}.

Despite their advantages, conventional PASS architectures face practical limitations. In particular, signals suffer from cumulative attenuation along the dielectric waveguide, and the requirement of physical line installation restricts deployment flexibility. These issues become more pronounced in large-scale or difficult-to-wire environments.
To overcome these limitations, wireless-fed PASS (Wi-PASS) architectures have been proposed, where a relay wirelessly delivers signals to the waveguide input \cite{wijewardhana2025}. {\color{black}However, the design of the wireless-feeding interface remains an important issue. On the one hand, antenna arrays can provide adaptive beamforming gain and are well suited for flexible transmission from the base station (BS) to the relay. On the other hand, employing antenna arrays at the relay may increase hardware complexity and provides limited passive spatial isolation in full-duplex operation. Conversely, horn antennas offer high directionality, narrow beamwidth, and favorable front-to-back radiation characteristics, which make them attractive as low-complexity relay-side receiving interfaces, but they lack the beam adaptability required at the BS.



Motivated by this tradeoff, we propose a \emph{hybrid BS-array/relay-horn architecture in this work}, where the BS employs an antenna array, while the relay adopts a highly directional horn antenna. This hybrid design is fundamentally different from existing approaches and is motivated by two key advantages. 
First, equipping the BS with an antenna array preserves beamforming adaptability, which is critical for practical wireless systems. In particular, the BS can dynamically adjust its transmission direction to serve users under varying spatial conditions, avoiding the performance degradation caused by fixed-beam transmission. 
Second, adopting a horn antenna at the relay significantly alleviates self-interference (SI) in full-duplex operation. Owing to its narrow beamwidth and high front-to-back ratio, the horn antenna provides strong inherent spatial isolation between transmit and receive paths \cite{Horn_antenna_book}. This structural property can reduce the SI level at the relay and complements practical cancellation techniques.
By combining these two properties, the proposed hybrid architecture achieves a favorable tradeoff between beam adaptability and interference mitigation, which is not attainable with conventional designs.}

Based on this architecture, we study the problem of minimizing the total system power consumption subject to a quality-of-service (QoS) constraint at the user equipment (UE). The optimization variables include the pinching-antenna position, the relay amplification gain, and the BS transmit power. By exploiting the structure of the end-to-end signal-to-noise ratio (SNR), we derive closed-form solutions for all optimization variables.

{\color{black}
The main contributions of this paper are as follows:
\begin{itemize}
\item We propose a hybrid BS-array/relay-horn Wi-PASS architecture that preserves beam adaptability at the BS while reducing residual SI through directional relay-horn reception.
\item We formulate a QoS-constrained total power minimization problem by jointly optimizing the pinching-antenna position, relay gain, and BS transmit power.
\item We derive closed-form expressions for the optimal pinching-antenna position, relay gain, and transmit power.
\item Numerical results demonstrate that the proposed hybrid architecture significantly improves power efficiency compared to conventional schemes.
\end{itemize}
}

{\color{black}The present work focuses on a fundamental single-user link-level design; multi-user Wi-PASS operation would require additional scheduling, resource allocation, and possibly multiple pinching antennas or waveguides.}

\begin{figure}[t]
\centering
{\includegraphics[width=0.42\textwidth]{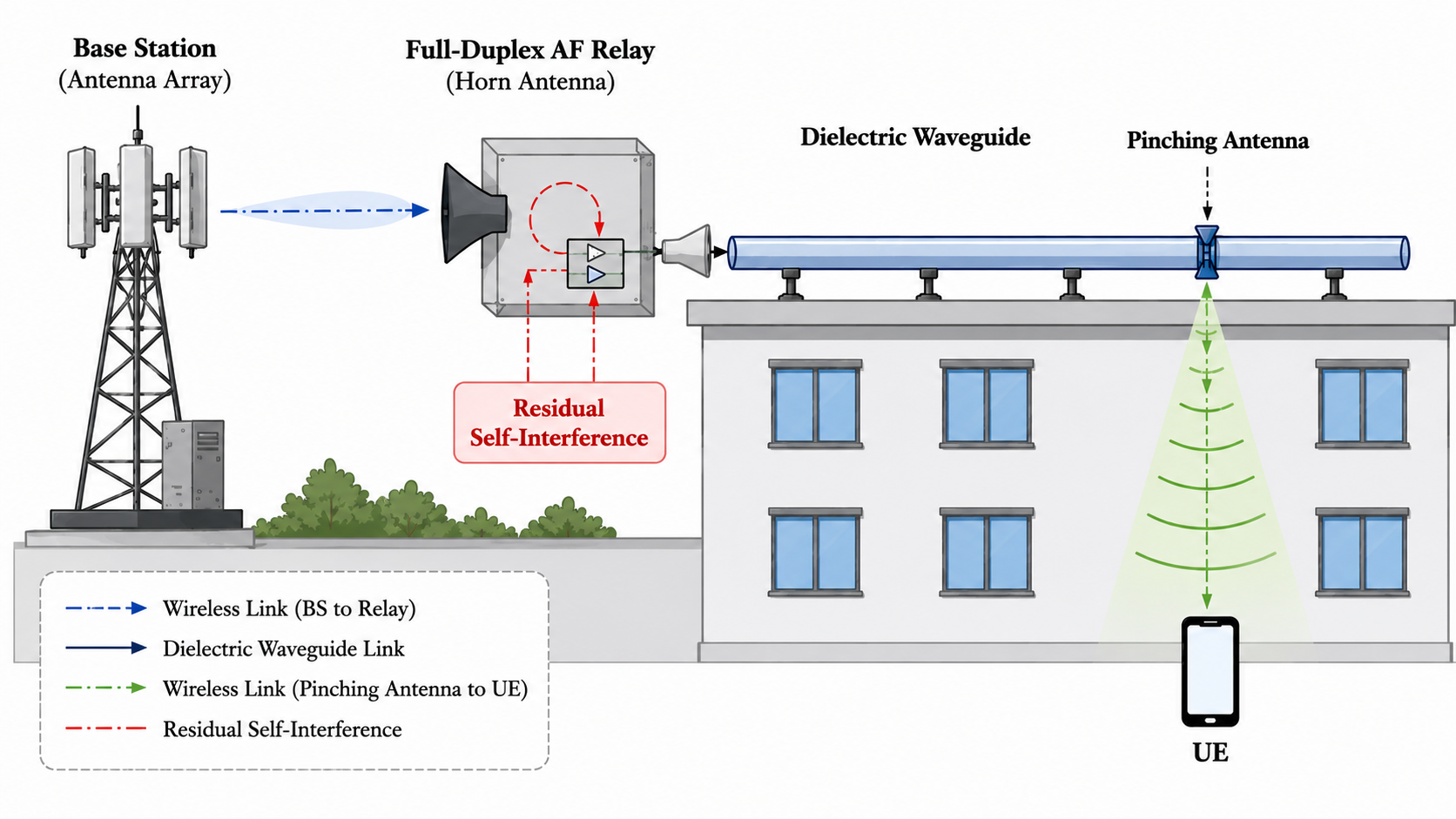}}
\caption{{\color{black}Proposed hybrid Wi-PASS system.}}
\label{fig:system model}
\end{figure}

\section{{\color{black}Proposed Hybrid Wi-PASS Architecture}}
{\color{black}As shown in Fig. \ref{fig:system model}, we propose a hybrid Wi-PASS for coverage extension in high-frequency communications. In such a system, the BS employs an antenna array to enable adaptive beamforming toward the relay,} while the relay uses a horn antenna for reception and feeds the signal into the dielectric waveguide, where it is radiated toward the UE via the pinching antenna. {\color{black}In practice, the relay horn should be placed close to the dielectric-waveguide input and oriented such that its main beam is aligned with the waveguide-feeding interface.}


\subsection{Signal Model}

Let $x$ denote the transmitted symbol from the BS with $\mathbb{E}[|x|^2]=1$. {\color{black}The BS applies a beamforming vector $\mathbf{w}\in \mathbb{C}^{N_t\times 1}$ satisfying $\|\mathbf{w}\|^2=1$, with $ N_t$ denoting the number of antennas at the BS. In full-duplex relaying, the relay receives the BS signal while simultaneously forwarding the amplified signal into the waveguide. After passive isolation and practical SI cancellation, the residual SI is modeled as an equivalent additive interference term, and thus, the received signal at the relay is given by}
\begin{equation}
{\color{black}y_{\rm{R}} = \sqrt{P_1}\mathbf{h}_{\rm{BR}}^H \mathbf{w}x + n_R+ z_{\mathrm{SI}},}
\end{equation}
where $P_1$ denotes the BS transmit power, $\mathbf{h}_{\rm{BR}}$ denotes the BS-to-relay channel vector, $n_R \sim \mathcal{CN}(0,\sigma_R^2)$ is the additive white Gaussian noise (AWGN) at the relay, {\color{black}and $z_{\mathrm{SI}}\sim\mathcal{CN}(0,\sigma_{\mathrm{SI}}^2)$ denotes the residual SI after cancellation.}

{\color{black}Define the effective first-hop channel as
\begin{equation}
g_1 \triangleq \mathbf{h}_{\rm{BR}}^H \mathbf{w}.
\end{equation}

For the considered line-of-sight (LoS)-dominant wireless-feeding link, the channel is modeled as
\begin{equation} \label{BS_relay}
\mathbf{h}_{\rm{BR}} = \sqrt{G_{\rm{r}} h_1}\,\mathbf{a}(\theta),
\end{equation}
where $G_{\rm{r}}$ denotes the relay horn antenna gain, $\mathbf{a}(\theta)$ is the BS array steering vector toward the relay}, and
\begin{equation}
h_1 = \frac{c^2}{16\pi^2 f^2 d_1^2}
\end{equation}
represents the free-space path loss, with $f$ denoting the carrier frequency, $c$ the speed of light, and $d_1$ the distance between the BS and the relay. {\color{black}By adopting matched beamforming toward the relay, the effective channel gain becomes
\begin{equation}
|g_1|^2 = N_t G_{\rm{r}} h_1.
\end{equation}
}
The relay applies an amplification factor $\beta$ and forwards the received signal into the dielectric waveguide. The transmitted signal at the relay is given by $ x_{\rm{R}} = \beta y_{\rm{R}}$.
Accordingly, the relay transmit power is\footnote{{\color{black}Following common system-level full-duplex modeling, the residual SI is treated as an impairment that degrades the effective relay observation and hence the end-to-end SNR, while the relay transmit-power budget is determined by the desired received signal and thermal noise components.}}
\begin{equation}
P_2 \triangleq \mathbb{E}[|x_{\rm{R}}|^2] = \beta^2\left(P_1|g_1|^2 + \sigma_R^2\right).
\end{equation}

The total relay power consumption is modeled as
\begin{equation}
P_{\mathrm{relay}} = \frac{P_2}{\eta_{\mathrm{PA}}} + P_{\mathrm{amp,circ}},
\end{equation}
where $\eta_{\mathrm{PA}}\in(0,1]$ denotes the power amplifier efficiency and $P_{\mathrm{amp,circ}}$ represents the circuit power consumption at the relay.

Let $g_2$ denote the effective channel between the pinching antenna and UE. Then, the received signal at the UE is
\begin{equation}
y_{\mathrm{UE}} = g_2 x_{\rm{R}} + n_{\mathrm{UE}},
\end{equation}
where $n_{\mathrm{UE}} \sim \mathcal{CN}(0,\sigma_{\mathrm{UE}}^2)$ denotes the AWGN at the UE.

{\color{black}
Accordingly, the end-to-end SNR at the UE is given by 
\begin{equation} \label{SNR}
\gamma =
\frac{P_1\beta^2 |g_1|^2 |g_2|^2}
{\sigma_{\mathrm{UE}}^2 + \beta^2 |g_2|^2 (\sigma_R^2+\sigma_{\mathrm{SI}}^2)}.
\end{equation}


It should be noted that the above model adopts a system-level communication-theoretic abstraction for analytical tractability. Detailed electromagnetic effects, such as horn antenna radiation patterns, polarization mismatch, alignment errors, and the exact coupling mechanisms contributing to SI, are not explicitly modeled. {\color{black}These factors may affect the practical link budget and deployment robustness,} and their aggregate impact is captured through effective channel gains and the residual SI term. A more detailed geometry- or electromagnetic-aware modeling is left for future work.
}

The pinching antenna is located at $\Phi_{\mathrm{Pin}} = (x_{\mathrm{Pin}}, 0, d)$, where $x_{\mathrm{Pin}} \in [0, L]$ and $L$ denotes the length of the dielectric waveguide. The UE is located at $\Phi_{\mathrm{UE}} = (x_{\mathrm{UE}}, y_{\mathrm{UE}}, 0)$. Following existing PASS studies, the relay-to-UE channel gain is modeled as \cite{tyrovolas2025}
\begin{equation} \label{relay_UE}
|g_2|^2 =
\frac{c^2 e^{-\alpha_D x_{\mathrm{Pin}}}}
{16\pi^2 f^2 \|\Phi_{\mathrm{UE}} - \Phi_{\mathrm{Pin}}\|^2},
\end{equation}
where $\alpha_D$ denotes the attenuation coefficient of the waveguide.

{\color{black}

The free-space terms in \eqref{BS_relay} and \eqref{relay_UE} are adopted to capture the dominant LoS components of the considered relay-based Wi-PASS. Specifically, the BS-to-relay link is a fixed point-to-point
wireless-feeding link, where the relay is deployed at a known location and the BS array can steer its beam toward the relay. For the relay-to-UE link, the relay and the waveguide shorten the effective over-the-air propagation distance, while the position-adjustable pinching antenna further helps establish a short-range LoS-like link by radiating from a favorable location along the waveguide. Therefore,
the adopted model provides a tractable first-order approximation for the proposed architecture. Practical impairments such as blockage, scattering, misalignment, {\color{black}radiation efficiency, PA-waveguide coupling efficiency,} and molecular absorption are not explicitly modeled and are left for future geometry-aware or measurement-based studies.

}

\subsection{Problem Formulation}
Our objective is to minimize the total power consumption of the system while guaranteeing a target QoS at the UE. Specifically, we jointly optimize the pinching antenna position, the BS transmit power, and the relay amplification gain. The optimization problem is formulated as\footnote{{\color{black}This work focuses on power minimization and does not explicitly model the impact of the reconfiguration latency of PASS on system performance. This assumption is suitable for quasi-static or slowly varying deployments, while latency-aware PASS optimization for highly dynamic scenarios is left for future work.}}
\begin{subequations} \label{P1}
\begin{align}
\min_{x_{\mathrm{Pin}},\,P_1,\,\beta} \quad
& P_1+ P_{\mathrm{relay}} \\
\mathrm{s.t.}\quad
& x_{\mathrm{Pin}} \in [0,L], \\
& \gamma \geq \gamma_0, \\
& P_1 \geq 0,\;\; \beta \geq 0,
\end{align}
\end{subequations}
where $\gamma_0$ denotes the minimum SNR requirement at the UE.


\section{Proposed Solution}

Substituting $P_{\mathrm{relay}}$ into the objective function yields
\begin{equation}
P_1 + P_{\mathrm{relay}}
= P_1 + \frac{\beta^2\left(P_1 |g_1|^2 + \sigma_R^2\right)}{\eta_{\mathrm{PA}}}
+ P_{\mathrm{amp,circ}}.
\end{equation}
Since $P_{\mathrm{amp,circ}}$ is constant, we consider the equivalent objective
\begin{equation}
\eta_{\mathrm{PA}} P_1 + \beta^2 \left(P_1 |g_1|^2 + \sigma_R^2\right).
\end{equation}

For fixed $(P_1,\beta)$, the SNR in \eqref{SNR} is monotonically increasing in $|g_2|^2$. Therefore, the optimal pinching antenna position is obtained by maximizing $|g_2|^2$, independently of $(P_1,\beta)$. 

\begin{theorem}
    The optimal position $x_{\mathrm{Pin}}^\star$ is given in closed form as
\begin{equation}
x_{\text{Pin}}^{\star} =
\begin{cases}
0,  
& \text{if}~ \Delta<0~ \text{or}~ x_{2} \leq 0  ~\text{or}~ x_{1} \geq 0 \\ 
& \text{and}~ f(0) \geq f(\min(L, x_{2})) \\[6pt]
\min(L, x_{2}), &\text{otherwise}
\end{cases}
\end{equation}
where $\Delta=4-4\alpha_D^2 (y_{\rm{UE}}^2+d^2 )$, $x_1=x_{\rm{UE}}- \frac{1+ \sqrt{1-\alpha_D^2(y_{\rm{UE}}^2+d^2 )  }}{\alpha_D}$, $x_2=x_{\rm{UE}}- \frac{1- \sqrt{1-\alpha_D^2(y_{\rm{UE}}^2+d^2 )  }}{\alpha_D}$, and $f(x)=\frac{e^{-\alpha_Dx}}{ (x_{\rm{UE}}-x)^2+ y_{\rm{UE}}^2 +d^2 }$. 
\end{theorem}
\begin{IEEEproof}
To determine the optimal pinching antenna position, we maximize $f(x)$ by equivalently maximizing its natural logarithm
\begin{equation}
g(x) = \ln f(x)
= -\alpha_D x - \ln\!\left((x_{\mathrm{UE}} - x)^2 + y_{\mathrm{UE}}^2 + d^2\right),
\end{equation}
which preserves the stationary points of $f(x)$.

Taking the derivative of $g(x)$ over $x$ and setting it to zero yields
\begin{equation}
-\alpha_D + \frac{2(x_{\mathrm{UE}} - x)}{(x_{\mathrm{UE}} - x)^2 + y_{\mathrm{UE}}^2 + d^2} = 0.
\end{equation}

Let $u = x_{\mathrm{UE}} - x$ and define $C = y_{\mathrm{UE}}^2 + d^2$. The above condition can be rewritten as
\begin{equation}
\alpha_D (u^2 + C) = 2u,
\end{equation}
which leads to the quadratic equation
\begin{equation}
\alpha_D u^2 - 2u + \alpha_D C = 0.
\end{equation}

When the discriminant $\Delta = 4 - 4\alpha_D^2 C$ is negative, the quadratic admits no real solution, and $f(x)$ is strictly decreasing over $x \in [0,L]$. Consequently, the maximum is attained at $x=0$.

Otherwise, the quadratic equation yields two real solutions
\begin{equation}
u = \frac{1 \pm \sqrt{1 - \alpha_D^2 (y_{\mathrm{UE}}^2 + d^2)}}{\alpha_D},
\end{equation}
which correspond to the stationary points
\begin{subequations}
\begin{align}
x_1 &= x_{\mathrm{UE}} - \frac{1 + \sqrt{1 - \alpha_D^2 (y_{\mathrm{UE}}^2 + d^2)}}{\alpha_D}, \\
x_2 &= x_{\mathrm{UE}} - \frac{1 - \sqrt{1 - \alpha_D^2 (y_{\mathrm{UE}}^2 + d^2)}}{\alpha_D}.
\end{align}
\end{subequations}

By examining the sign of the derivative, it follows that $x_1$ corresponds to a local minimum, whereas $x_2$ corresponds to a local maximum of $f(x)$.

If $x_2 \leq 0$, $f(x)$ is strictly decreasing on $[0,L]$, and the maximum occurs at $x=0$. If $x_1 \geq 0$, $f(x)$ decreases on $[0,x_1]$, and the optimal solution is determined by comparing $f(0)$ and $f(\min(L,x_2))$. Finally, when $x_1 < 0$ and $x_2 > 0$, $f(x)$ increases on $[0,x_2]$, and the maximum is attained at $x=\min(L,x_2)$. This completes the proof.
\end{IEEEproof}


With the optimal pinching antenna position fixed, the original problem reduces to a joint optimization over the BS transmit power $P_1$ and the relay gain $\beta^2$, which can be formulated as {\color{black}
\begin{align}
\min_{P_1,\beta^2} \quad 
& J(P_1,\beta^2) = \eta_{\mathrm{PA}}P_1 + \beta^2\left(P_1|g_1|^2 + \sigma_R^2\right) \\
\mathrm{s.t.} \quad 
& \frac{P_1\beta^2 |g_1|^2 |g_2|^2}
{\sigma_{\mathrm{UE}}^2 + \beta^2|g_2|^2 \tilde{\sigma}_R^2} \ge \gamma_0,
\end{align}
where $\tilde{\sigma}_R^2 = \sigma_R^2 + \sigma_{\mathrm{SI}}^2$ denotes the effective relay impairment.}

\begin{theorem}
The minimum value of $J(P_1,\beta^2)$ is given by{\color{black}
\begin{equation}
J^\star =
\frac{\eta_{\mathrm{PA}}\gamma_0 \tilde{\sigma}_R^2}{|g_1|^2}
+
\frac{\gamma_0 \sigma_{\mathrm{UE}}^2}{|g_2|^2}
+
\frac{2 \sigma_{\mathrm{UE}}}{|g_1||g_2|}
\sqrt{\eta_{\mathrm{PA}}\gamma_0\left(\gamma_0 \tilde{\sigma}_R^2 + \sigma_R^2\right)}.
\end{equation}
}
\end{theorem}

\begin{proof}
Since both the objective function $J(P_1,\beta^2)$ and the end-to-end SNR monotonically increase with respect to $P_1$ and $\beta^2$, the minimum is achieved when the SNR constraint is satisfied with equality, i.e., {\color{black}
\begin{equation} \label{gamma_0}
\gamma_0 =
\frac{P_1\beta^2 |g_1|^2 |g_2|^2}
{\sigma_{\mathrm{UE}}^2 + \beta^2 |g_2|^2 \tilde{\sigma}_R^2}.
\end{equation}

Rearranging \eqref{gamma_0} yields
\begin{equation}
\beta^2 =
\frac{\gamma_0 \sigma_{\mathrm{UE}}^2}
{|g_2|^2\left(P_1|g_1|^2 - \gamma_0 \tilde{\sigma}_R^2\right)}.
\end{equation}


Substituting the above expression for $\beta^2$ into the objective function yields
\begin{equation}
J(P_1) = \eta_{\mathrm{PA}}P_1 +
\frac{\gamma_0 \sigma_{\mathrm{UE}}^2 \left(P_1|g_1|^2 + \sigma_R^2\right)}
{|g_2|^2\left(P_1|g_1|^2 - \gamma_0 \tilde{\sigma}_R^2\right)}.
\end{equation}

Define the auxiliary variable
\begin{equation}
u \triangleq P_1|g_1|^2 - \gamma_0 \tilde{\sigma}_R^2,
\end{equation}
which implies $P_1|g_1|^2 = u + \gamma_0 \tilde{\sigma}_R^2$. Then, the objective becomes
\begin{equation} \label{Ju_31}
J(u) =
\eta_{\mathrm{PA}}\frac{u + \gamma_0 \tilde{\sigma}_R^2}{|g_1|^2}
+
\frac{\gamma_0 \sigma_{\mathrm{UE}}^2\left(u + \gamma_0 \tilde{\sigma}_R^2 + \sigma_R^2\right)}
{|g_2|^2 u}.
\end{equation}

Ignoring constant terms independent of $u$ in \eqref{Ju_31}, we obtain
\begin{equation}
J(u) =
\eta_{\mathrm{PA}}\frac{u}{|g_1|^2}
+
\frac{\gamma_0 \sigma_{\mathrm{UE}}^2\left(\gamma_0 \tilde{\sigma}_R^2 + \sigma_R^2\right)}
{|g_2|^2 u}
+ \text{const}.
\end{equation}

Taking the derivative over $u$ and setting it to zero yields
\begin{equation}
\eta_{\mathrm{PA}}\frac{1}{|g_1|^2}
-
\frac{\gamma_0 \sigma_{\mathrm{UE}}^2\left(\gamma_0 \tilde{\sigma}_R^2 + \sigma_R^2\right)}
{|g_2|^2 u^2}
= 0.
\end{equation}

Solving gives
\begin{equation}
u^\star =
\frac{|g_1| \sigma_{\mathrm{UE}}}{|g_2|\sqrt{\eta_{\mathrm{PA}}}}
\sqrt{\gamma_0\left(\gamma_0 \tilde{\sigma}_R^2 + \sigma_R^2\right)}.
\end{equation}

Substituting back yields the optimal BS transmit power
\begin{equation}
P_1^\star =
\frac{\gamma_0 \tilde{\sigma}_R^2}{|g_1|^2}
+
\frac{\sigma_{\mathrm{UE}}}{|g_1||g_2|}
\sqrt{\frac{\gamma_0\left(\gamma_0 \tilde{\sigma}_R^2 + \sigma_R^2\right)}{\eta_{\mathrm{PA}}}}.
\end{equation}

The optimal relay gain is
\begin{equation}
\beta^{2\star} =
\frac{\sigma_{\mathrm{UE}}}{|g_1||g_2|}
\sqrt{
\frac{\eta_{\mathrm{PA}}\gamma_0}
{\gamma_0 \tilde{\sigma}_R^2 + \sigma_R^2}
}.
\end{equation}

Substituting $(P_1^\star,\beta^{2\star})$ into the objective function yields $J^\star$. This completes the proof.}
\end{proof}

{\color{black}
Finally, the minimum total power consumption is given by
\begin{equation}
P_{\mathrm{total}}^\star =
\frac{\gamma_0 \tilde{\sigma}_R^2}{|g_1|^2}
+
\frac{\gamma_0 \sigma_{\mathrm{UE}}^2}{\eta_{\mathrm{PA}}|g_2|^2}
+
\frac{2 \sigma_{\mathrm{UE}}}{|g_1||g_2|}
\sqrt{
\frac{\gamma_0\left(\gamma_0 \tilde{\sigma}_R^2 + \sigma_R^2\right)}{\eta_{\mathrm{PA}}}
}
+
P_{\mathrm{amp,circ}}.
\end{equation}

It is observed that the presence of residual SI increases the effective relay impairment $\tilde{\sigma}_R^2$, thereby requiring higher transmit power to satisfy the QoS constraint.}

\section{Numerical Results}

In this section, numerical results are provided to evaluate the performance of the proposed hybrid Wi-PASS. Unless otherwise specified, the carrier frequency is $f=28$ GHz, the system bandwidth is 400 MHz, the noise figure is 10 dB, the waveguide attenuation coefficient is $\alpha_D=0.01$, the horn antenna gain is 20 dBi, the power amplifier efficiency is $\eta_{\mathrm{PA}}=0.9$, and the relay circuit power is $P_{\mathrm{amp,circ}}=0.2$ W. The UE is uniformly distributed within a $30$ m $\times$ $10$ m coverage area, while the waveguide length and height are set to $L=30$ m and $d=3$ m, respectively. The BS--relay distance is $d_1=50$ m and the target SNR is $\gamma_0=20$ dB. {\color{black}The BS is equipped with $N_t=16$ antennas and adopts analog beamforming with one RF chain. The RF chain power consumption is 0.1 W, and each phase shifter consumes 0.01 W. For the proposed scheme, the residual SI power at the horn-based relay is set to $\sigma_{\mathrm{SI}}^2=-85$ dBm.} {\color{black}These values are representative system-level parameter settings used for relative performance comparison; their exact values may vary with hardware implementation and waveguide design.}

To demonstrate the effectiveness of the proposed scheme, we consider the following benchmark schemes.

\emph{{\color{black}Direct transmission}:} The BS directly serves the UE using analog beamforming with the same number of antennas, without relay assistance, dielectric waveguide, or pinching antenna. The BS–UE channel follows a non-line-of-sight (NLoS) model with path loss exponent 3.5 and 8 dB log-normal shadowing \cite{samimi2014characterization}.

\emph{Array with PASS:} This benchmark has the same Wi-PASS structure as the proposed scheme, except that the relay uses a 16-element receive array with single-stream dominant singular-vector beamforming. Its residual SI power is set to be 10 dB higher than that of the horn-based relay {\color{black}as a representative setting} reflecting weaker passive isolation.

\emph{Array without PASS:} This benchmark is similar to the array-based relay scheme, but the relay-to-UE link is modeled as free-space transmission from a fixed relay antenna position, without the position-adjustable pinching antenna.

\emph{Hybrid without PASS:} This benchmark retains the same BS-array/relay-horn wireless-feeding architecture as the proposed scheme, but removes the position-adjustable pinching antenna and uses a fixed relay-to-UE free-space link.


\begin{figure}[t]
\centering
{\includegraphics[width=0.43\textwidth]{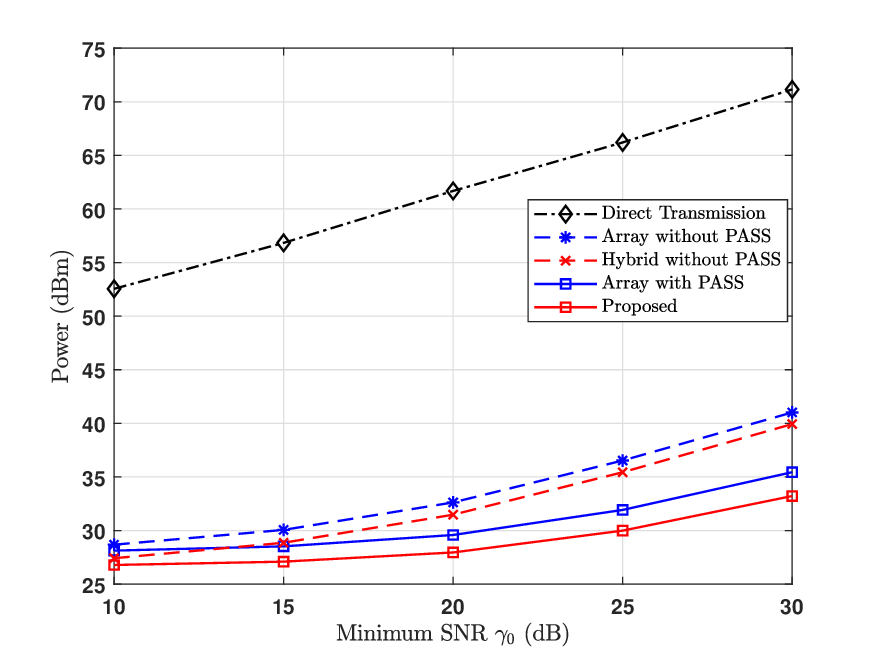}}
\caption{Total power consumption versus $\gamma_0$.}
\label{fig:power_vs_gamma}
\end{figure}

{\color{black}Fig.~\ref{fig:power_vs_gamma} shows the total power consumption (in dBm) versus the target SNR $\gamma_0$. As expected, the total power consumption of all schemes increases with $\gamma_0$. The direct transmission benchmark consumes the highest power because the direct BS--UE link suffers from severe NLoS propagation loss and shadowing at high frequency. Compared with the array-based relay schemes, the proposed scheme achieves lower power consumption by exploiting the high directional gain and reduced residual SI of the horn antenna at the relay. Moreover, the proposed scheme also outperforms the hybrid scheme without PASS, which demonstrates the benefit of optimizing the pinching-antenna position. This confirms that the proposed design jointly benefits from directional wireless feeding and flexible radiation along the waveguide.

Fig.~\ref{fig:power_vs_d1} plots the total power consumption versus the BS--relay distance $d_1$. For the relay-assisted schemes, increasing $d_1$ weakens the BS--relay wireless-feeding link, thereby requiring higher transmit power to satisfy the QoS constraint. For the direct-transmission benchmark, the power consumption also increases with $d_1$ because the BS--UE distance increases accordingly in the considered geometry, which further aggravates the high-frequency NLoS path loss and shadowing. The proposed scheme consistently achieves the lowest power consumption among the relay-assisted schemes, benefiting from both the directional relay-horn reception and the optimized pinching-antenna position. Compared with the schemes without PASS, the performance gain of the proposed scheme confirms that position-adjustable radiation along the waveguide can effectively reduce the relay-to-UE propagation loss. Moreover, the gap between the proposed scheme and the array-based PASS benchmark highlights again the advantage of the relay horn antenna in reducing residual SI and providing directional gain.

Fig.~\ref{fig:power_vs_SI} illustrates the impact of residual SI at the relay. Note that the residual SI level of the array-based relay schemes is set 10 dB higher than that of the horn-based schemes to reflect their weaker passive isolation. It can be observed that the power consumption of all full-duplex relay-assisted schemes increases as the residual SI becomes stronger, which is consistent with Theorem 2, where the residual SI increases the effective relay impairment $\tilde{\sigma}_R^2$. In contrast, the direct transmission benchmark remains unchanged since it does not involve a full-duplex relay and is thus independent of relay SI. Among the relay-assisted schemes, the proposed one is the most robust to SI, owing to the lower residual SI enabled by the relay horn antenna and the additional channel improvement provided by optimized PASS. The sharp increase of the array-based schemes at high SI levels further confirms the importance of relay-side passive isolation in practical full-duplex Wi-PASS.}

\begin{figure}[t]
\centering
{\includegraphics[width=0.43\textwidth]{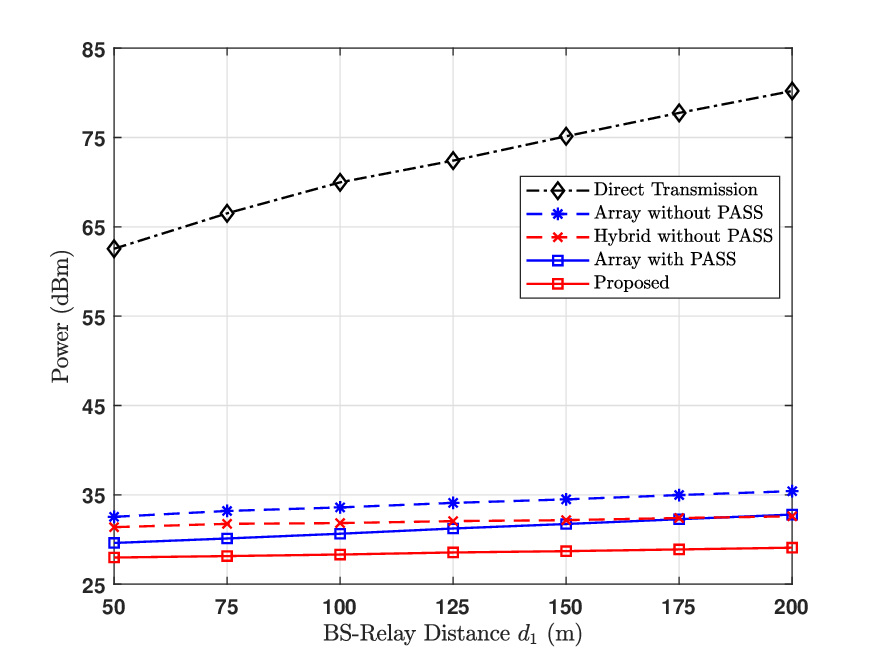}}
\caption{Total power consumption versus $d_1$.}
\label{fig:power_vs_d1}
\end{figure}

\begin{figure}[t]
\centering
{\includegraphics[width=0.43\textwidth]{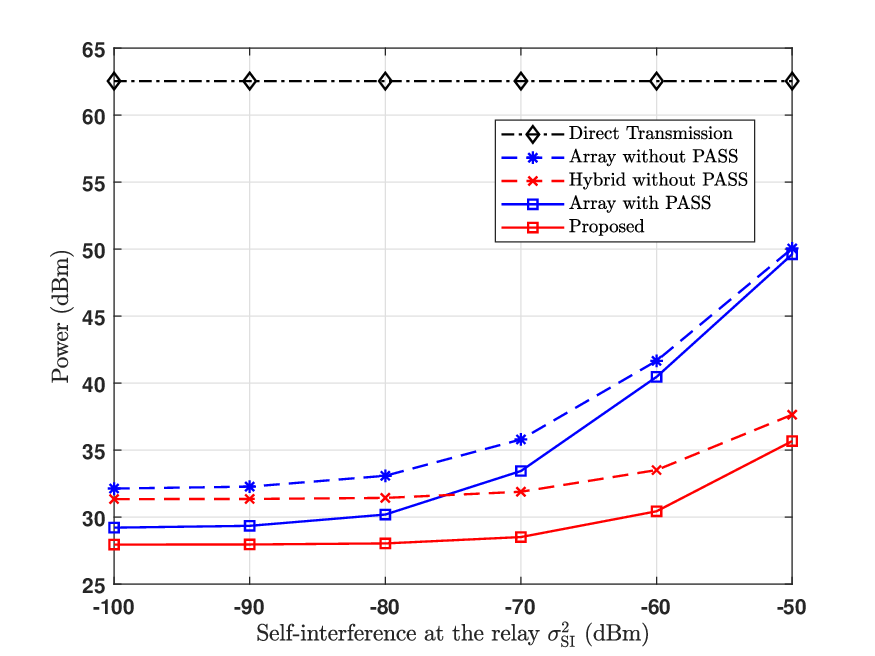}}
\caption{{\color{black}Total power consumption versus residual SI at the relay.}}
\label{fig:power_vs_SI}
\end{figure}

\section{Conclusion} 
\label{Sec:Conclusion}
In this paper, we proposed a hybrid Wi-PASS for high-frequency coverage extension, where a BS antenna array enables adaptive directional transmission and a relay horn antenna supports directional waveguide feeding. Residual SI was incorporated into the end-to-end SNR model, based on which a QoS-constrained total power minimization problem was formulated by jointly optimizing the pinching-antenna position, BS transmit power, and relay gain. Closed-form solutions were derived, revealing the joint impact of waveguide attenuation, free-space propagation loss, and residual SI on power consumption. Numerical results demonstrated that the proposed scheme reduces total power consumption compared with direct transmission, array-based relay schemes, and hybrid transmission without PASS, while improving robustness against residual SI. 
{\color{black}Future work will extend this framework to multi-user scenarios with joint scheduling, resource allocation, and multiple-PA or multiple-waveguide configurations, as well as measurement-based hardware calibration and deployment-complexity/hardware-cost evaluation.}


\bibliographystyle{IEEEtran}
\bibliography{biblio}

@ARTICLE{ding2024,
  author={Ding, Zhiguo and others},
  journal={IEEE Trans. Commun.},  
  title={Flexible-Antenna Systems: A Pinching-Antenna Perspective}, 
  year={2025},
  volume={73},
  number={10},
  pages={9236-9253},
  month={Oct.},
  }

@ARTICLE{Atsushi_22,
  author={Atsushi Fukuda and others},
  journal={Technical Journal}, 
  title={Pinching antenna using a dielectric
waveguide as an antenna}, 
  year={2022},
  volume={23},
  number={3},
  pages={5-12},
  month={Jan.},
  }

@misc{fu2025,
      title={Power Minimization for NOMA-assisted Pinching Antenna Systems With Multiple Waveguides}, 
      author={Yaru Fu and others},
      year={arxiv.org/abs/2503.20336, 2025},
      eprint={2503.20336},
      archivePrefix={arXiv},
      primaryClass={cs.IT},
}

@ARTICLE{Zeng_COMML25,
  author={Zeng, Ming and others},
  journal={IEEE Wirel. Commun. Lett.}, 
  title={Sum Rate Maximization for {NOMA}-Assisted Uplink Pinching-Antenna Systems}, 
  year={2026},
  volume={15},
  number={},
  pages={280-284},
  }

@ARTICLE{Sun_TVT18,
  author={Sun, Shu and others},
  journal={IEEE Trans. Veh. Tech.}, 
  title={Propagation Models and Performance Evaluation for {5G} Millimeter-Wave Bands}, 
  year={2018},
  volume={67},
  number={9},
  pages={8422-8439},
  month={Sep.},}

@ARTICLE{Hao_Network22,
  author={Hao, Wanming and Zhou, Fuhui and Zeng, Ming and Dobre, Octavia A. and Al-Dhahir, Naofal},
  journal={IEEE Network}, 
  title={Ultra Wideband {THz} {IRS} Communications: Applications, Challenges, Key Techniques, and Research Opportunities}, 
  year={2022},
  volume={36},
  number={6},
  pages={214-220},
  }

@ARTICLE{yang2025,
  author={Yang, Zheng and others},
  journal={IEEE Wirel. Commun.}, 
  title={Pinching Antennas: Principles, Applications and Challenges}, 
  year={2025},
  month={Oct.},
  volume={},
  number={},
  pages={1-10},
}

@ARTICLE{liu2025pinching,
  author={Liu, Yuanwei and others},
  journal={IEEE Commun. Mag.}, 
  title={Pinching-Antenna Systems: Architecture Designs, Opportunities, and Outlook}, 
  year={2025},
  month={Jan.},
  volume={},
  number={},
  pages={1-7},
  }

@ARTICLE{tyrovolas2025,
  author={Tyrovolas, Dimitrios and others},
  journal={IEEE Trans. Cogn. Commun. Netw.}, 
  title={Performance Analysis of Pinching-Antenna Systems}, 
  year={2026},
  volume={12},
  number={},
  pages={590-601},
  month={Apr.},}

@ARTICLE{zeng2025EE,
  author={Zeng, Ming and others},
  journal={IEEE Wirel. Commun. Lett.}, 
  title={Energy-Efficient Resource Allocation for {NOMA}-Assisted Uplink Pinching-Antenna Systems}, 
  year={2025},
  volume={14},
  number={11},
  pages={3695-3699},
  month={Nov.},}

@misc{zeng2025_WCM,
      title={Resource Allocation for Pinching-Antenna Systems: State-of-the-Art, Key Techniques and Open Issues}, 
      author={Ming Zeng and others},
      year={2025},
      eprint={2506.06156},
      archivePrefix={arXiv},
      primaryClass={cs.IT},
      url={https://arxiv.org/abs/2506.06156}, 
}

@book{Horn_antenna_book,
  title={Antenna Theory: Analysis and Design},
  author={Constantine A Balanis},
  year={2016},
  publisher={John wiley \& sons}
}

@ARTICLE{Xiao_COMML25,
  author={Xiao, Jian and others},
  journal={IEEE Commun. Lett.}, 
  title={Channel Estimation for Pinching-Antenna Systems {(PASS)}}, 
  year={2025},
  volume={29},
  number={8},
  pages={1789-1793},
  month={Aug.}
 }

@ARTICLE{Zhao_TCOM25,
  author={Zhao, Jingjing and others},
  journal={IEEE Trans. Commun.}, 
  title={Pinching-Antenna Systems-Enabled Multi-User Communications: Transmission Structures and Beamforming Optimization}, 
  year={2025},
  month={Dec.},
  volume={},
  number={},
  pages={1-1},
  }

@ARTICLE{Ouyang_COMML25,
  author={Ouyang, Chongjun and others},
  journal={IEEE Commun. Lett.}, 
  title={Array Gain for Pinching-Antenna Systems ({PASS})}, 
  year={2025},
  volume={29},
  number={6},
  pages={1471-1475},
  month={Jun.},
  }

@ARTICLE{wijewardhana2025,
  author={Wijewardhana, Kasun R. and others},
  journal={IEEE Commun. Mag.}, 
  title={Wireless-Fed Pinching-Antenna Systems ({Wi-PASS}) for {NextG} Wireless Networks}, 
  year={2026},
  volume={},
  number={},
  pages={1-7},
  doi={10.1109/MCOM.001.2500663}}

@article{samimi2014characterization,
  title={Characterization of the 28 {GHz} millimeter-wave dense urban channel for future {5G} mobile cellular},
  author={Samimi, Mathew K and Rappaport, Theodore S},
  journal={NYU Wireless TR},
  volume={1},
  pages={1--322},
  year={2014}
}

\balance

\end{document}